\date{}
\title{Multilevel polynomial partitions
and simplified range searching\thanks{Research
supported by the  ERC Advanced Grant No.~267165.
}
}
\newif\ifcmts
\cmtstrue

\author{
{\sc Ji\v{r}\'{\i} Matou\v{s}ek}
\\
   {\footnotesize Department of Applied Mathematics}\\[-1.5mm]
   {\footnotesize  Charles University, Malostransk\'{e} n\'{a}m. 25}\\[-1.5mm]
{\footnotesize  118~00~~Praha~1,
   Czech Republic, and}\\
{\footnotesize    Department of Computer Science}\\[-1.5mm]
{\footnotesize    ETH Zurich,
      8092 Zurich, Switzerland}
\and
{\sc Zuzana Pat{\'a}kov{\'a}\thanks{Partially supported by the project CE-ITI
(GACR P202/12/G061) of the Czech Science Foundation and by the Charles University Grants SVV-2014-260103 and GAUK 690214.}}
\\
   {\footnotesize Department of Applied Mathematics and} \\[-1.5mm]
   {\footnotesize Computer Science Institute}\\[-1.5mm]
   {\footnotesize  Charles University, Malostransk\'{e} n\'{a}m. 25}\\[-1.5mm]
{\footnotesize  118~00~~Praha~1,
   Czech Republic}
}

\documentclass[11pt]{article}

\usepackage{a4wide}
\usepackage{latexsym}
\usepackage{amssymb,amsmath,amsthm}
\usepackage{graphicx}
\usepackage{url}
\usepackage{hyperref}
\usepackage{color}
\usepackage{enumerate}

\newtheorem{theorem}{Theorem}[section]

\newtheorem{lemma}[theorem]{Lemma}
\newtheorem{corol}[theorem]{Corollary}
\newtheorem{conj}[theorem]{Conjecture}

\newcommand{\heading}[1]{\vspace{1ex}\par\noindent{\bf\boldmath #1}}

\makeatletter
\newcommand{\ProofEndBox}{{\ifhmode\unskip\nobreak\hfil\penalty50 \else
          \leavevmode\fi\quad\vadjust{}\nobreak\hfill$\Box$
            \finalhyphendemerits=0 \par}}
\makeatother

\newcommand{\R}{{\mathbb{R}}}
\newcommand{\C}{{\mathbb{C}}}
\newcommand{\Q}{{\mathbb{Q}}}

\newcommand{\PC}{{\mathrm{P}\C}}  

\newcommand\eps{\varepsilon}

\newcommand\FF{\mathcal{F}}
\newcommand\GG{\mathcal{G}}
\newcommand\TT{\mathcal{T}}
\newcommand\NF{\mathsf{NF}}
\newcommand\supp{\mathsf{supp}}
\newcommand\lm{\ensuremath{\ell m}}

\definecolor{seda}{rgb}{.7,.7,.7}

\def\:{\colon}

\long\def\onefigure#1#2{
\begin{figure*}[tbp]
\begin{center}
#1
\end{center}
\caption{#2}
\end{figure*}
}

\def\immediateFigure#1{%
\smallskip\begin{center}#1\end{center}\smallskip }

\newcommand{\labfig}[2]  
{\onefigure{\mbox{\includegraphics{Figures/#1}}}{\label{f:#1} #2} }

\newcommand{\labfigw}[3]  
{\onefigure{\mbox{\includegraphics[width=#2]{Figures/#1}}}{\label{f:#1} #3}}

\newcommand{\immfig}[1]  
{\immediateFigure{\mbox{\includegraphics{Figures/#1}}}}

\newcommand{\immfigw}[2] 
{\immediateFigure{\mbox{\includegraphics[width=#2]{Figures/#1}}}}

\ifcmts
\newcommand{\marrow}{\marginpar{\boldmath$\longleftarrow$}}
\newcommand{\jirka}[1]{\ifhmode\newline\fi\marrow \textsf{*** (JM: ) #1\newline}}
\newcommand{\zuzka}[1]{\ifhmode\newline\fi\marrow \textsf{*** (ZS: ) #1\newline}}

\else
\newcommand{\marrow}{}
\newcommand{\jirka}[1]{}

\fi

\begin{document}

\maketitle

\begin{abstract} 
The polynomial partitioning method of Guth and Katz
[arXiv:1011.4105] has numerous applications in discrete and 
computational geometry. It partitions a given $n$-point set $P\subset\R^d$
using the zero set $Z(f)$ of a suitable $d$-variate polynomial $f$.
Applications of this result are often complicated by the problem,
what should be done with the points
of $P$ lying within $Z(f)$? A natural approach is to partition
these points with another polynomial and continue further 
in a similar manner.  So far it has been pursued with limited success---several
authors managed to construct and apply a second 
partitioning polynomial, but further progress has been prevented
by technical obstacles.

We provide a polynomial partitioning method with up to $d$ polynomials
in dimension $d$, which allows for a complete decomposition
of the given point set. We apply it to obtain a new
algorithm for the \emph{semialgebraic range searching problem}.
Our algorithm has running time bounds similar to a recent algorithm
by Agarwal, Sharir, and the first author [SIAM~J.~Comput. 
42: 2039--2062, 2013], but it is simpler both conceptually and
technically. 
While this paper has been in preparation, 
Basu and Sombra, as well as Fox, Pach, Sheffer, Suk, and Zahl, obtained
results concerning polynomial partitions which overlap with ours
to some extent.
\end{abstract}

\section{Introduction}

\heading{Polynomial partitions. } Since the late 1980s,
numerous problems in discrete and computational geometry have been
solved by geometric divide-and-conquer method, where a suitable partition
of space is used to subdivide a geometric problem into simpler subproblems.

The earliest, and most widely applied, kinds of such partitions are 
\emph{cuttings}, based mainly on ideas of Clarkson (e.g., \cite{c-narsc-87}) 
and  Haussler and Welzl \cite{hw-ensrq-87}.  
See, e.g., \cite{Chazelle-cuttings} for a survey of cuttings and
their applications. 

Using cuttings as the main tool, another kind of space partition,
called \emph{simplicial partitions},
was introduced in \cite{m-ept-92} (and further improved 
by Chan \cite{Chan-newpart}). Given an $n$-point set $P\subset \R^d$
and a parameter $r>1$, a simplicial $\frac1r$-partition
is a collection of simplices (of dimensions $0$ through $d$),
such that each of them contains at most  $n/r$ points of~$P$
and together they cover $P$. In 
Chan's version, they can also be assumed to be pairwise disjoint.

Let us introduce  the following convenient terminology:
a set $A$ \emph{crosses} a set $B$ if $A$ intersects $B$ but does not
contain it. 
The main parameter of a simplicial partition is the maximum number
of simplices of the partition that can be simultaneously crossed by
a hyperplane (or, equivalently, by a halfspace). 
One can construct simplicial partitions where
this number is bounded by $O(r^{1-1/d})$ \cite{m-ept-92,Chan-newpart},
which is asymptotically optimal in the worst case
(throughout this paper, we consider the space dimension
$d$ as a constant, and the implicit constants in asymptotic notation
may depend on it, unless explicitly stated otherwise). 

Simplicial partitions work mostly fine for problems involving points
and hyperplanes in $\R^d$. However, they are much less useful if
hyperplanes are replaced by lower-dimensional objects---such as lines---or curved objects---such as spheres---or other hypersurfaces. 

Guth and Katz \cite{GK2} invented a new kind of partitions, called
\emph{polynomial partitions}, which overcome these drawbacks to some extent.
The most striking application of polynomial partitions so far is probably
still the original one in \cite{GK2} in a solution of Erd\H{o}s'
problem of distinct distances (also see Guth
\cite{Guth-GK2-simplif} for a simplified but weaker version
of the main result of \cite{GK2}),
but a fair number of other applications
have been found since then: see Solymosi and Tao \cite{SoTa},
Zahl \cite{zahl3d}, Kaplan et al. \cite{KMS},
Kaplan et al. \cite{KMSS}, Zahl \cite{zahl-complSzT}, Wang et al.
 \cite{WangYangZhang},
 Agarwal et al. 
\cite{AMSII}, Sharir, Sheffer, Zahl \cite{ShaSheZa}, and
Sharir and Solomon \cite{ShaSo} (our list is most likely incomplete
and we apologize for omissions).


Given an $n$-point set $P\subset\R^d$ and a parameter $r>1$, we say that
a nonzero polynomial $f\in\R[x_1,\ldots,x_d]$ is a 
\emph{$\frac1r$-partitioning polynomial} for $P$ if none of the
connected components of $\R^d\setminus Z(f)$ contains more than
$n/r$ points of~$P$.

Guth and Katz \cite{GK2} proved that, for every $P$
and every $r>1$, there exists a $\frac1r$-partitioning polynomial
of degree $O(r^{1/d})$. 

From the results of real algebraic geometry on
the complexity of arrangements of zero sets of polynomials
(see \cite{BasuPollackRoy-book})
it follows that any hyperplane $h$ intersects at most $O(r^{1-1/d})$
components of $\R^d\setminus Z(f)$, and hence any halfspace
crosses at most $O(r^{1-1/d})$ components of $\R^d\setminus Z(f)$.
Moreover, using a more recent result of Barone and Basu \cite{BB}
discussed below, one obtains that  an algebraic variety $X$
of dimension $k$ defined by polynomials of 
constant-bounded degrees crosses at most $O(r^{k/d})$ 
components of $\R^d\setminus Z(f)$. 
In this respect, polynomial partitions match the performance of simplicial
partitions concerning hyperplanes and give a crucial advantage for
other varieties. However, they still leave an important issue open: namely,
what should be done with the \emph{exceptional set} $P^*:=P\cap Z(f)$
that ends up lying within the zero set of the partitioning polynomial.

\heading{Multilevel polynomial partitions.}
At first sight, it may seem that this issue can be remedied, say,
by a suitable perturbation of the polynomial
$f$. However, if all of $P$ lies on
a line in $\R^d$, say, then a degree-$D$ polynomial can partition
it into at most $D+1$ pieces, and so if we want all of $P$ to be 
partitioned into pieces of size $n/r$, then we will need degree about $r$,
as opposed to $r^{1/d}$ in the Guth--Katz polynomial partition theorem.

A natural idea is to partition the exceptional set $P^*$ further by another
polynomial $g$ such that $Z(f,g):=Z(f)\cap Z(g)$ has dimension at most $d-2$.
If $Z(f,g)$ again contains many points of $P^*$, we would like
to partition them further by a third polynomial $h$ with
$\dim Z(f,g,h)\le d-3$, and so on.

This program encounters several technical difficulties,
and so far it has been realized only up to the second partitioning
polynomial $g$ in \cite{zahl3d} and \cite{KMSS} (also see
\cite{zahl-complSzT}). 

Our main result is the following multilevel partition theorem.

\begin{theorem}\label{t:multilev} 
For every integer $d>1$ there is a 
constant $K$ such that the following hold. Given
an $n$-point set $P\subset \R^d$ and a parameter $r>1$, there 
are numbers $r_1,r_2,\ldots,r_d\in [r,r^K]$, positive integers $t_1,t_2, \ldots, t_d$, a partition
\[
P=P^*\cup \bigcup_{i=1}^d \bigcup_{j=1}^{t_i} P_{ij}
\]
of $P$ into disjoint subsets, and for every $i,j$,
a connected set  $S_{ij}\subseteq\R^d$ containing $P_{ij}$,
such that $|P_{ij}|\le n/r_i$ for all $i,j$, $|P^*|\le r^K$,
and the following hold:
\begin{enumerate}
\item[\rm(i)] If $h\in\R[x_1,\ldots,x_d]$ is a polynomial of degree
bounded by a constant $D_0$, and $X=Z(h)$ is its zero set, then,
 for every $i=1,2,\ldots,d$, the number of the $S_{ij}$ crossed 
by $X$ is at most $O\left(r_i^{1-1/d}\right)$,
with the implicit constant also depending on~$D_0$.
\item[\rm(ii)]
If  $X$ is an algebraic variety in $\R^d$ of dimension
at most $k\le d-2$ defined by polynomials 
of degree bounded by a constant $D_0$, 
then, for every $i=1,2,\ldots,d$, the number of the $S_{ij}$ crossed
by $X$ is bounded by $O\left(r_i^{1-1/(k+1)}\right)$.
\end{enumerate}
\end{theorem}

We will need only part (i), while part (ii) is stated for possible future
use, since it can be handled with very little extra work.

\heading{Related work. } 
The problem concerning the exceptional set $P^*$
in a single-level polynomial partition has been
addressed in various ways in the literature.

In one of the theorems in Agarwal et al. \cite{AMSII},  $P^*$ is forced to
be at most of a constant size, by an infinitesimal perturbation of $P$.
However, this strategy cannot be used
in incidence problems, for example, 
where a perturbation destroys the structure
of interest. Moreover, for algorithmic purposes, known methods
of infinitesimal perturbation are applicable with a reasonable
overhead only for constant values of~$r$.

Solymosi and Tao \cite{SoTa} handle the exceptional set
 essentially by projecting it to a hyperplane. This yields
a $(d-1)$-dimensional problem,  which is handled recursively.
Their method allows them to deal only with constant values of $r$, 
and consequently it
yields bounds that are suboptimal by factors of $n^\eps$ (where $\eps>0$
is arbitrarily small but fixed number). 

Another variant of the strategy  of projecting $P^*$ to a hyperplane
was used in  \cite{AMSII}; there $r$ could be chosen as a small
but fixed power of $n$, leading to only polylogarithmic extra factors,
as opposed to $n^\eps$ with constant~$r$. However, the resulting algorithm
and proof are complicated, since one has to keep track
of several parameters and solve a tricky recursion.

Our proof of Theorem~\ref{t:multilev} also involves a projection trick,
but the projection is encapsulated in the proof and simple to analyze, and 
in applying the theorem we can work in the original space all the time.

In this paper we apply an algorithmic enhancement of Theorem~\ref{t:multilev},
stated below, to recover the main result of Agarwal et al.~\cite{AMSII}
in a way that is simpler both conceptually and technically.


While this paper was in preparation, two groups of researchers
announced results concerning multilevel polynomial partitions,
which partially overlap with ours. 
Fox, Pach, Sheffer, Suk, and Zahl \cite{fox-pach-sheffer-suk-zahl} as well as Basu and Sombra \cite{basu-sombra}, 
obtained results similar to our key lemma (Lemma~\ref{l:key}),
but with different proofs. However, the
Basu--Sombra result  works just for varieties of codimension two and hence it cannot be used for our range searching algorithm.
On the other hand, Fox et al. have no restriction on the dimension of the variety, but they have to assume the variety is irreducible. 
The important feature of our method is that we are able to avoid computing irreducible components which is crucial from algorithmic point of view. 
For more details we refer to the discussion in Section \ref{s:irred}.


\heading{Range searching with semialgebraic sets. }
Here we consider a basic and long-studied question in computational
geometry. 

Let $P$ be a set of $n$ points in $\R^d$
and let $\Gamma$ be a family of geometric ``regions,''
called \emph{ranges}, in $\R^d$. For example, $\Gamma$ can be the set of all
axis-parallel boxes, balls, simplices, or cylinders, or the set of all
intersections of pairs of ellipsoids. In the \emph{$\Gamma$-range searching}
problem, we want to preprocess $P$ into a data structure so that
the number of points of $P$ lying in a query range $\gamma \in \Gamma$ can be
counted efficiently. More generally, we may be given a weight function 
on the points in $P$ and we ask for the cumulative weight of the 
points in $P\cap \gamma$ (our result applies in this more general
setting as well). We consider the \emph{low-storage} variant of
$\Gamma$-range searching, where the data structure  is allowed
to use only linear or near-linear storage, and the goal is to
make the query time as small as possible.

We study \emph{semialgebraic range searching}, where $\Gamma$ is a set
of constant-complexity semialgebraic sets. We recall 
that a \emph{semialgebraic set} is a subset of $\R^d$ obtained
from a finite number of sets of the form $\{x\in\R^d\mid g(x)\ge 0\}$, 
where $g$ is a $d$-variate polynomial with integer coefficients,
by Boolean operations (unions, intersections,
and complementations). Specifically, let $\Gamma_{d,D,s}$ denote
the family of all semialgebraic sets in $\R^d$ defined
by at most $s$ polynomial inequalities of degree at most $D$ each.
By \emph{semialgebraic range searching} we mean
$\Gamma_{d,D,s}$-range searching for some parameters $d,D,s$.

This problem and various special cases of it have been studied
in many papers. We refer to \cite{ae-grsr-97,m-grs-95} for background
on range searching and to \cite{AMSII} for a more detailed discussion
of the problem setting and previous work.

The main result of \cite{AMSII} is as follows.

\begin{theorem}\label{t:large-r}
Let $d,D_0,s$, and $\eps>0$ be constants.
Then the $\Gamma_{d,D_0,s}$-range searching problem for an arbitrary
$n$-point set in $\R^d$
can be solved with $O(n)$ storage, $O\left(n^{1+\eps}\right)$ expected
preprocessing time, and $O\left(n^{1-1/d}\log^B n\right)$ query time, where
$B$ is a constant depending on $d,D_0,s$ and~$\eps$.
\end{theorem}

As announced, here we provide a new and simpler proof. 
%
%
Basically we apply Theorem~\ref{t:multilev}, but for the algorithmic
application, we need to amend it with an algorithmic part, essentially
asserting that the construction in Theorem~\ref{t:multilev} can be
executed in time depending polynomially on $r$ and linearly on $n$
(we again stress that $d$ is taken as a constant). Moreover, we
need that the $S_{ij}$ can be handled algorithmically---they are
semialgebraic sets of controlled complexity. 
We will use the real RAM model of computation where
we can compute exactly with arbitrary real numbers and each arithmetic operation is executed in
unit time.

A precise statement is as follows.

\begin{theorem}[Algorithmic enhancement of Theorem~\ref{t:multilev}]
\label{t:algo-multilev}
Given 
$P\subset \R^d$ and $r$ as in Theorem~\ref{t:multilev}, 
one can compute the sets $P^*$, $P_{ij}$, and $S_{ij}$ in time $O\left(nr^C\right)$,
where $C=C(d)$ is a constant. Moreover, for every $i$, the number
$t_i$ of the $P_{ij}$ is $t_i=O\left(r^C\right)$, and each $S_{ij}$ is a semialgebraic
 set 
defined by at most $O\left(r^C\right)$ polynomial inequalities of maximum degree
$O\left(r^C\right)$. For every $i=1,2,\ldots,d$,
every range $\gamma\in \Gamma
_{d,D_0,s}$ crosses at most
$O\left(r_i^{1-1/d}\right)$
of the $S_{ij}$, with the constant of proportionality
depending on $d,D_0,s
$.
\end{theorem}

\section{Algebraic preliminaries}\label{s:prelim}

Throughout the paper we assume that we are working in the 
\emph{Real RAM} model of
computation, where arithmetic operations with arbitrary real numbers
can be performed exactly and in unit time. This is the most usual
model in computational geometry.

We could also consider the bit model (a.k.a. Turing machine model),
assuming the input points rational or, say, algebraic. Then the
analysis would be more complicated, but we believe that,
with sufficient care, bounds analogous to those we obtain
in the Real RAM model can be derived as well, with an extra multiplicative
term polynomial in the bit size of the input numbers.
For example, the algorithms of real algebraic geometry we use
are also analyzed in the bit model in \cite{BasuPollackRoy-book},
and the polynomiality claims we rely on still hold. However,
at present we do not consider this issue sufficiently important
to warrant the additional complication of the paper.

\heading{Notions and tools from algebraic geometry over $\C$.}
A \emph{real algebraic variety} $V$ is a subset of some $\R^d$
that can be expressed as $V=Z(f_1,\ldots,f_m)$,
i.e., the set of common zeros 
of finitely many polynomials  $f_1, \ldots, f_m \in 
\R[x_1,\ldots,x_d]$. For a \emph{complex algebraic variety},
$\R$ is replaced with $\C$ (the complex numbers).\footnote{More precisely,
these are \emph{affine} algebraic varieties, while other kinds
of algebraic varieties, such as projective or quasiprojective
ones, are often considered in the literature. Here, with a single
exception, we suffice with the affine case.}

As in the introduction, we will use $Z(f)$ for the real zeros of
a (real) polynomial $f\in\R[x_1,\ldots,x_d]$, while 
$Z_\C(f)$ is the set of all zeros of a complex or real polynomial
in $\C^d$. For a real polynomial $f$ we have $Z(f) = Z_\C(f) \cap \R^d$.

A nonempty complex variety $V$ is called \emph{irreducible}  
if it cannot be written as the union of two proper complex subvarieties,
and similarly for real varieties.
The empty set is not considered to be irreducible.
Note that $Z(f)$ can be irreducible over $\R$ even if $Z_\C(f)$ is reducible
over $\C$.
An easy example is the variety $V(x^2+y^2)$.
It is well known that every nonempty variety can be uniquely decomposed 
into a finite number of irreducible components, none containing another.

For a complex variety $V$, we will use the 
notions of \emph{dimension} $\dim V$ and \emph{degree} $\deg V$.
These can be defined in several equivalent ways. We
refer to the literature such as \cite{CLO,harris,hartshorne}, for rigorous treatment.
Here we just recall a rather intuitive definition and state the
properties we will actually use.

The dimension of $V\subseteq\C^d$ can be defined
as the largest $k$ such that a generic
$(d-k)$-dimensional complex affine subspace $F$ of $\C^d$ intersects $V$
in finitely many points, and the degree is the number
of intersections (which is the same for all generic $F$).
To explain the meaning of ``generic'', let us consider only
the subspaces $F=F(a)$ that can be expressed by the equations
$x_{i+d-k}=a_{i0}+\sum_{j=1}^{d-k} a_{ij}x_j$, $i=1,\ldots,k$,
for some $a=(a_{ij})_{i=1}^k{}_{j=0}^{d-k}\in \C^{k(d-k+1)}$.
The $F(a)$ being generic means that the point
$a$ does not lie in the zero set of a certain nonzero polynomial
(depending on $V$).
In particular,  almost all subspaces $F$ in the sense of measure are generic.
We note that the dimension of $\C^d$ is $d$ and its degree is 1.

If $V=Z_\C(f)$ is the zero set of a single squarefree polynomial, then 
$\deg V=\deg f$. We will always assume that the polynomials
we deal with are squarefree.

For a real algebraic variety $V$, the definition with a generic
affine subspace does not quite make sense, and in real algebraic
geometry, the dimension is usually defined, for the
more general class of semialgebraic sets,  as the
largest $k$ such that $V$ contains the image
of a $k$-dimensional open cube under an injective
semialgebraic map; see \cite{BochnakCosteRoy,BasuPollackRoy-book}.
An equivalent way of defining the dimension of a real algebraic
variety $V$ uses the Krull dimension\footnote{The Krull dimension
of a ring $R$ is the largest $n$
such that there exists a chain $I_0\subsetneq I_1\subsetneq\cdots
\subsetneq I_n$ of nested prime ideals in~$R$.} 
of the \emph{coordinate ring} $\R[x_1,\ldots, x_d]/I(V)$, 
where $I(V)$  is the ideal
of all real polynomials vanishing on $V$;
 see 
\cite[Cor.~2.8.9]{BochnakCosteRoy} for this equivalence.
For complex case the dimension defined via generic affine subspaces
coincides with the Krull dimension of the coordinate ring $\C[x_1,\ldots, x_d]/I_\C(V)$;
see \cite[Chapter 11]{harris}.

We will need the following fact, which is apparently standard
(for example, it is mentioned without proof as Remark~13 in \cite{newBB}),
although so far we have not been able to locate an explicit
reference (Whitney \cite[Lemma~8]{whit-realvar} proves a similar statement,
but he uses definitions that are not standard in the current literature).

\begin{lemma}\label{l:dims}
 Let $V \subseteq \C^d$ be a complex
variety. Then $V\cap\R^d$ is a real variety and
$\dim (V \cap \R^d) \leq \dim V$.
\end{lemma}

This is perhaps not as obvious as it may seem, because
if we identify $\C^d$ with $\R^{2d}$ in the usual way,
then topologically, a $k$-dimensional complex variety $V$
has (real) dimension~$2k$.

\begin{proof}[Sketch of proof]
If $V=Z_\C(f_1,\ldots,f_m)$ for $f_1,\ldots,f_m\in\C[x_1,\ldots,x_d]$, then 
\[
V\cap\R^d=Z(f_1\overline f_1,\ldots,f_m\overline f_m),
\]
where the bar denotes complex conjugation. Each
$f_i\overline f_i$ is a real polynomial, and so $V\cap\R^d$
is a real variety.

The inequality for the dimensions can be checked, for example,
by employing the definition of the dimensions via the Hilbert
function (see, e.g.,  \cite{CLO}), 
which is well known to be equivalent to the Krull dimension
definition. Indeed, if $f \in \C[x_1,\ldots,x_d]$ is a complex polynomial 
of degree at most $D$ vanishing
on $V$, we can write $f=f_1+if_2$,  where $f_1, f_2 \in \R[x_1,\ldots,x_d]$ correspond to the real and complex parts of coefficients of $f$, respectively. Then $\deg f_1$ and $\deg f_2$ are at most $D$
and both $f_1$ and $f_2$ vanish on $V\cap\R^d$.
Therefore, if $(g_1,\ldots,g_m)$ is
a basis of the real vector space of all real polynomials
of degree at most $D$ vanishing on $V\cap\R^d$, then
the $g_1,\ldots,g_m$, regarded as complex polynomials, generate
the complex vector space of all complex polynomials
of degree at most $D$ vanishing on~$V$. It follows that the Hilbert function
of the complex variety 
$V$ is at least as large as the Hilbert function
of the real variety~$V\cap\R^d$.
\end{proof}

\begin{lemma}[A generalized B\'ezout inequality]\label{l:gbezout}
Let $V\subseteq\C^d$ be an irreducible variety, let $f\in\C[x_1,\ldots,x_d]$
be a polynomial that does not vanish identically on $V$, and let
$W_1,\ldots,W_k$ be the irreducible components of $V\cap Z_\C(f)$.
Then all of the $W_i$ have dimension $\dim (V)-1$, and their degrees
satisfy
\[
\sum_{i=1}^k\deg W_i\le \deg(V)\deg(f).
\]
\end{lemma}

\begin{proof} We may assume that $f$ is irreducible
(if not, we decompose it into irreducible factors, use the lemma 
for each factor separately, and add up the degrees).

The first part about dimension of every irreducible component is exactly \cite[Exercise~I.1.8]{hartshorne} (also see \cite[Prop.~I.7.1]{hartshorne}).

As for the statement with degrees,
we let $\overline V\subseteq \PC^d$ be the projective closure
of $V$, and similarly for $\overline{Z_\C(f)}$. Let
$Y_1,\ldots,Y_m$ be the irreducible components of $\overline V\cap
\overline{Z_\C(f)}$. By \cite[Thm.~I.7.7]{hartshorne}, we have
$\sum_{i=1}^m\deg Y_i\le\deg(\overline V)\deg(\overline{Z_\C(f)})=
\deg(V)\deg(f)$. For every $W_i$, the projective closure $\overline W_i$
is irreducible, and so it equals a unique $Y_{j(i)}$,
and $\deg W_i\le \deg Y_{j(i)}$. The lemma follows.
Also see \cite[Thm.~1]{heintz-defi} for a similar statement.
\end{proof}

We will need to apply the lemma to a variety that is not necessarily
irreducible. We will use that the degree is additive in the following sense:
if $V_1,\ldots,V_k$ are the irreducible components of
a variety $V$, with $\dim V_i=\dim V$ for all $i$,
then $\deg V=\sum_{i=1}^k\deg V_i$.

We also need the property that a variety of degree $\Delta$ can be defined
by polynomials of degree at most~$\Delta$.

\begin{theorem}[Prop.~3 in \cite{heintz-defi}]\label{t:heintz-def}
 Let $V$ be an irreducible affine variety in $\C^d$. Then there exist $d+1$ 
polynomials $f_1,\ldots,f_{d+1} \in \C[x_1,\ldots,x_d]$
 of degree at most $\deg V$ such that $V=Z_\C(f_1,\ldots,f_{d+1})$.
\end{theorem}

\heading{Ideals and Gr\"obner bases.}  
For polynomials $f_1,\ldots,f_m\in\C[x_1,\ldots,x_d]$,
the \emph{ideal} $I$ generated by $f_1,\ldots,f_m$ is the
set of all polynomials of the form $h_1f_1+\cdots+h_mf_m$,
$h_1,\ldots,h_m\in \C[x_1,\ldots,x_d]$. Every such ideal
has a \emph{Gr\"obner basis}, which is a set of polynomials 
that also generates $I$ and has certain favorable properties;
see, e.g., \cite{CLO} for an introduction. 

Each Gr\"obner basis is associated with a certain \emph{monomial ordering}.
We will use only Gr\"obner bases with respect to a
\emph{lexicographic ordering}, where monomials in the variables
$x_1,\ldots,x_d$ are first ordered according to the powers of $x_d$,
then those with the same power of $x_d$ are ordered according to
powers of $x_{d-1}$, etc.  In other words, we consider lexicographic ordering w.r.t. $x_d > x_{d-1} > \cdots > x_1$.




We will need the following theorem:
\begin{theorem}\label{t:groebner}
Assuming $d$ fixed and
given polynomials $f_1,\ldots,f_m\in\C[x_1,\ldots,x_d]$ with $\deg f_i \geq 1$,
a Gr\"obner basis of the ideal generated by the $f_i$
can be computed in time polynomial in $\sum_{i=1}^m\deg f_i$.
\end{theorem}

We have not found an explicit reference in the literature that
would provide Theorem \ref{t:groebner}. 
In particular, for the usual Buchberger algorithm and variations of it,
only much worse bounds seem to be known.
However, Theorem \ref{t:groebner}
follows  by inspecting the method of K\"uhnle and Mayr
\cite{kuhnleMayr}  for finding a Gr\"obner basis in exponential space. (Also see \cite{MayrRitscher} for a newer algorithm.)

Before providing the details, we need one definition:
For any polynomial $h \in \C[x_1,\ldots,x_d]$, the \emph{normal form} $\NF(h)$ w.r.t. $I \subseteq \C[x_1,\ldots,x_d]$
is the unique \emph{irreducible}\footnote{A polynomial $h$ is \emph{reducible} w.r.t. $I$ if $\supp(h) \cap \langle\lm(I)\rangle \neq \emptyset,$ 
where the support of $h$ is a set of all monomials occurring in $h$ (i.e., having nonzero coefficient)
and $\langle\lm(I)\rangle = \langle\lm(f) \colon f \in I\rangle$ is an ideal of all \emph{leading monomials} of $I$,
where leading monomial $\lm(f)$ is the largest monomial occurring in $f$. } polynomial w.r.t. $I$ in the coset\footnote{$h+I= \{h+f\colon f \in I\}.$} $h + I$. 
Recall that we have fixed lexicographic ordering.\footnote{We note that the algorithm by \cite{kuhnleMayr} requires the monomial ordering given by rational weight matrix.
The weight matrix of lexicographic ordering consists just of zero's and one's, and hence it is rational. See \cite{kuhnleMayr} for details.}
We note that K\"uhnle and Mayr work over the field $\Q$, 
however, the theoretical background works also for $\C$. 
Let $I \subseteq \C[x_1,\ldots, x_d]$ be an ideal whose Gr\"obner basis we want to compute and assume it is 
generated by  $m$ polynomials of degree bounded by $D$.  

\begin{enumerate}
 \item[(i)] First important lemma \cite[Section 5]{kuhnleMayr},\cite[Lemma 3]{MayrRitscher} is 
that the \emph{reduced} Gr\"obner basis  is always equal to the set of all the polynomials $h - \NF(h)$, where $h$ is a monomial \emph{minimally reducible}\footnote{A monomial $h$ is \emph{minimally reducible}
w.r.t. $I$ if it is reducible w.r.t. $I$ but none of its proper divisors is reducible w.r.t. $I$.} w.r.t. $I$. 

\item[(ii)] Let $h \in \C[x_1,\ldots,x_d]$ be arbitrary but fixed. Our next goal is to compute $\NF(h)$ w.r.t. $I$. Since $h - \NF(h) \in I$, there is a representation
 \begin{equation}\label{eq:hermann}
  h-\NF(h) = \sum_{i=1}^mc_if_i \quad \text{with} \quad c_1,\ldots,c_m \in \C[x_1,\ldots,x_d]. 
 \end{equation} 
The next step is to rewrite the polynomial equation (\ref{eq:hermann}) to a system of linear equations.
Recall that $h$ and  $f_i$'s are fixed and $\NF(h)$ and  $c_i$'s are unknowns.
Let us assume that $\deg c_i \leq E$ for all $i$ and some $E$.
Expanding all the polynomials $h, f_i, c_i$ and also the polynomial $r:=\NF(h)$ to sums of monomials
 and comparing the coefficients of left and right side in (\ref{eq:hermann}), we get one linear equation for every term.
 If $\deg \NF(h) \leq N$ for some $N$, it can be shown that there are at most $(\max(N, D+E))^d$ equations in no more than 
 $N^d+mE^d$ unknowns. It follows that all these linear equations can be rewritten into a single matrix equation and the size of the matrix is bounded by $N^d + m(D+E)^d$. For more details we refer to \cite[Section 3]{kuhnleMayr}.
 Note  that it can happen that there are more unknowns than equations. 
 Fortunately, since we are interested in a solution with minimal $r$ (w.r.t. lexicographic ordering), we can always decrease the number of unknowns by 
 putting the coefficient corresponding to the largest monomial of $r$ to be zero. 
 For more details (and also example) we again refer to \cite[Section 3]{kuhnleMayr}.


 \item[(iii)] Now we want to bound degrees of $c_i$'s and also the degree of $\NF(h)$. By Hermann \cite{hermann, mayrMeyer}, the degrees of $c_i$'s are bounded by 
$E := \deg(h- \NF(h))+(mD)^{2^d}.$
Dub\'e \cite{dube}  showed the
existence of a Gr\"obner basis $G$ for $I$ where the degrees of all
polynomials in $G$ are bounded by $M:=2(D^2/2 + D)^{2^{d-1}}$.
Using this bound, K\"uhnle and Mayr \cite[Section 2]{kuhnleMayr} showed 
that the degree of the normal form of $h$ w.r.t. $I$ can be always bounded by $N:=((M+1)^d\deg(h))^{d+1}.$

\item[(iv)] It follows that to compute reduced Gr\"obner basis of $I$ it is enough to enumerate all monomials up
to Dub\'e's bound and calculate their normal forms and normal forms of all its direct divisors.
This can be done by solving the system of linear equations described in (ii).
\end{enumerate}

In order to turn the described method into an algorithm, we have to be able to efficiently solve a system of linear equations.
K\"uhnle and Mayr used Turing machines, that is why they need to work over $\Q$. Since we work with the Real RAM model of computation
which allows arithmetic operations with arbitrary real numbers (in unit time),
 we can use the described algorithm over~$\C$ as well.

Now we are ready to prove Theorem \ref{t:groebner}.

\begin{proof}[Proof of Theorem \ref{t:groebner}]

Clearly $D \leq \sum_{i=1}^m\deg f_i$ and $m \leq \sum_{i=1}^m\deg f_i$, since $\deg f_i \geq 1$ for every $i$.
It follows from (i)--(iv) that, for $d$ fixed, the Gr\"obner basis can be computed in polynomial time in $\sum_{i=1}^m\deg f_i$.
Indeed, by (ii) and (iii), the normal form of a polynomial of degree bounded by $O(D)$ can be computed in time polynomial in $D$, and hence also in 
$\sum_{i=1}^m\deg f_i$.
According to (iv), the step (ii) is repeated polynomially many times; the claim follows.
\end{proof}


\heading{Tools from real algebraic geometry. }
Let $\FF\subset\R[x_1,\ldots,x_d]$ be a finite set of polynomials. The
\emph{arrangement} of (the zero sets of) $\FF$  
is the partition
of $\R^d$ into maximal relatively open connected subsets,
called \emph{cells}, such that for each cell $C$
there is a subset $\FF_C\subseteq \FF$ such that
$C\subseteq Z(f)$ for all $f\in\FF_C$ and 
$C\cap Z(f)=\emptyset$ for all $f\in\FF\setminus\FF_C$.

Similar to \cite{AMSII},
a crucial tool for us is the following theorem of Barone and Basu.

\begin{theorem}[{Barone and Basu~\cite{BB}}]
\label{t:basu}
Let $V$ be a $k$-dimensional algebraic variety in $\R^d$
defined by a finite set $\FF$  of $d$-variate real polynomials, 
each of degree at most 
$D$, and let $\GG$ be a set of $s$ polynomials of degree at most
$E\ge D$.  Then  the number of those cells
of the arrangement of the zero sets of $\FF\cup\GG$ 
that are contained in $V$
is bounded by $O(1)^d D^{d-k}(sE)^k$.
\end{theorem}

We will be using the theorem only for $d$ a constant and
$\GG=\{g\}$ consisting of a single
polynomial to get an upper bound of $O(D^{d-k} E^k)$
on the number of connected components of $V\setminus Z(g)$.

For the range searching algorithm, we also need the following
algorithmic result on the construction of arrangements.

\begin{theorem}[{Basu, Pollack and Roy~\cite[Thm.~16.18]{BasuPollackRoy-book}}]
\label{t:make-arrg}
Let $\FF=\{f_1,\ldots,f_m\}$ be a set of $m$ real
$d$-variate polynomials, each of degree
at most~$D$. Then the arrangement of the zero sets of $\FF$ in $\R^d$
has at most $(mD)^{O(d)}$ cells, and it can be computed
in time at most $T= m^{d+1}D^{O(d^4)}$.
Each cell is described as a semialgebraic set
using at most $T$ polynomials of degree bounded by
$D^{O(d^3)}$. Moreover, the algorithm supplies adjacency information
for the cells, indicating which cells are contained in the boundary
 of each cell,
and it also supplies an explicitly given point in each cell.
\end{theorem}

\section{A key lemma: Partitioning Polynomial that does not vanish
on a variety}\label{s:key}

In this section we establish the following lemma, which will
allow us to deal with the exceptional sets and iterate
the construction of a partitioning polynomial. Although
we are dealing with a problem in $\R^d$, it will be more convenient
to work with complex varieties.
This is because algebraic varieties over an algebraically closed field
have some nice properties that fail for real varieties in general.

\begin{lemma}[Key lemma]\label{l:key}
Let $V\subseteq \C^d$ be a
complex algebraic variety of dimension $k \geq 1$, 
such that all of its irreducible components $V_j$ have
 dimension $k$ as well. Let $Q\subset V\cap\R^d$ be a
finite point set, and let $r>1$ be a parameter.  Then there
exists a real $\frac1r$-partitioning polynomial $g$ for $Q$
of degree at most $D=O(r^{1/k})$ that does not vanish identically on any of the irreducible components $V_j$ of $V$.
\end{lemma}

 Note that the bound on $\deg g$ in the key lemma cannot be improved to $O\left(\left(\frac{r}{\Delta}\right)^{1/k}\right)$, where $\Delta$ is the degree of $V$, unless there are some restrictive conditions on $r$.
We thank to the anonymous referee, who pointed it out. The example is as follows:
let us assume that all points of $Q$ lie on a $k$-flat $F$. Since $F$ is isomorphic to $\R^k$, 
partitioning of $Q$ corresponds to a partititioning in $\R^k$.
It is clear that if $V$ is formed by a union of $F$ and  many other $k$-flats parallel to $F$, 
then $\left(\frac{r}{\Delta}\right)^{1/k}$ can be made arbitrarily close to zero and hence $O\left(\left(\frac{r}{\Delta}\right)^{1/k}\right)$ cannot serve as a degree bound for a partitioning polynomial. 

However, we believe that, for an irreducible variety, one can hope for a better bound and we propose the following conjecture:

\begin{conj}
 Let $V\subseteq \C^d$ be an irreducible
complex algebraic variety of dimension $k \geq 1$ and degree $\Delta$. Let $Q\subset V\cap\R^d$ be a
finite point set, and let $r \geq \Delta^{k+1}, r>1$ be a parameter.  Then there
exists a real $\frac1r$-partitioning polynomial $g$ for $Q$
of degree at most $D=O\left(\left(\frac{r}{\Delta}\right)^{1/k}\right)$ that does not vanish identically on $V$.
\end{conj}

Note that for $k=d$ the affirmative answer follows from the partitioning theorem by Guth and Katz \cite{GK2}, and for $k=d-1$ from the 
theorem by Kaplan et al. \cite{KMSS} (for $d=3$) and also by Zahl \cite{zahl3d}.
We also note that Basu and Sombra propose similar conjecture, see \cite[Conj. 3.4]{basu-sombra}.

Even if the conjecture is true, we cannot use it for our range searching application unless we know how to effectively decompose a variety into irreducibles.



Before proving the key lemma, we first sketch the idea. The proof is based on a projection trick.
Let us consider the \emph{standard projection}
$\pi_d\:\C^d\to \C^{d-1}$ given by 
$(a_1,\ldots,a_d)\mapsto (a_1,\ldots,a_{d-1})$,
i.e., forgetting the last coordinate.
The standard projection of an affine variety need not be a variety in general
(consider, e.g., the projection of the hyperbola $Z(xy-1)$
on the $x$-axis). However, for every variety of dimension at most $d-1$,
there is a simple linear change of coordinates in $\C^d$  (Lemma \ref{l:coordinate})
after which the image of $V$ under the standard projection
is a variety in $\C^{d-1}$ (Theorem \ref{t:proj}). Moreover, this projection preserves
the dimension of the variety  (Theorem \ref{t:proj}).

The idea of the proof of the key lemma is to project
the given $k$-dimensional complex variety $V$ 
onto $\C^k$, by iterating the standard projection,
and, if necessary, coordinate changes in such a way that 
the image of $V$ is all of $\C^k$  (Corollary \ref{c:thatproj}).
Then we find a $\frac1r$-partitioning polynomial for the projection
of the given point set $Q$
by the Guth--Katz method, and we pull it back to a $\frac1r$-partitioning 
polynomial in $\R^d$.

We now present this approach in more detail.
We begin with a well-known sufficient condition guaranteeing
that  the standard projection of a variety is a variety
of the same dimension.

\begin{theorem}[{Projection theorem}]\label{t:proj}
 Let $I \subset \C[x_1,\ldots,x_d]$ be an ideal, $d \geq 2,$ and let 
$J := I \cap \C[x_1,\ldots,x_{d-1}]$ be the ideal 
consisting of all polynomials in $I$ that do not contain the
variable $x_d$.
 Suppose that $I$ contains a nonconstant polynomial $f$,
with $D=\deg f\ge 1$,
in which the monomial $x_d^D$ appears with a nonzero coefficient. Let $V=V(I)$ be a complex variety defined as the zero locus of 
all polynomials in $I$.
Then the image $\pi_d(V)$ under the standard projection
$\pi_d\:\C^d\to \C^{d-1}$ 
is the variety $Z_\C(J) \subseteq \C^{d-1}$,
and $\dim \pi_d(V)=\dim V$.
 \end{theorem}

\begin{proof}  Theorem~1.68 in \cite{DP} contains everything
in the theorem except for the claim $\dim \pi_d(V)=\dim V$.
For this claim, which is also standard, we first observe that,
for every point $a\in \pi_d(V)$, the $x_d$-coordinates of these preimages are 
roots of the nonzero univariate polynomial $f_a(x_d):=f(a_1,\ldots,a_{d-1},x_d)$.
In other words the extension $\C[x_1,\ldots,x_{d-1}]/J\subseteq \C[x_1,\ldots,x_d]/I$
is integral.\footnote{A ring $S$ is an \emph{integral extension}
of a subring $R\subseteq S$  if all elements
of $S$ are roots of monic polynomials in $R[x]$.}
 By  \cite[Thm.~2.2.5]{celistvost},
integral extension preserves  the (Krull) dimension. 
\end{proof}

%

The next standard lemma (a simple form of the \emph{Noether normalization} for infinite fields)
 implies that the condition in the projection theorem can always 
be achieved by a suitable change of coordinates. See, e.g., \cite[Lemma~1.69]{DP}.

\begin{lemma}\label{l:coordinate}
Let $f \in \C[x_1, \ldots, x_d]$ be a 
polynomial of degree $D \geq 1$. Then there 
are  coefficients $\lambda_1,\ldots,\lambda_{d-1}$ such that 
\[
f'(x_1,\ldots,x_d):=f(x_1+\lambda_1x_d,\ldots,x_{d-1}+\lambda_{d-1}x_d,x_d)
\]
is a polynomial of degree $D$ in which the monomial $x_d^D$ has
a nonzero coefficient. This holds for a generic choice of the
$\lambda_i$, meaning that there is a nonzero polynomial
$g\in\C[y_1,\ldots,y_{d-1}]$ such that $f'$ satisfies the condition
above whenever $g(\lambda_1,\ldots,\lambda_{d-1})\ne 0$.
Consequently, the condition on $f'$ holds for almost all
choices of a \emph{real} vector $(\lambda_1,\ldots,\lambda_{d-1})$.
\end{lemma}

By combining the projection theorem with Lemma~\ref{l:coordinate}
and iterating, we obtain the following consequence:

\begin{corol}\label{c:thatproj}
Let $V\subset\C^d$ be a complex variety of dimension $k$,
$1\le k\le d-1$, for which all irreducible
components also have dimension $k$. 
Then there is a linear map $\pi\:\C^d\to\C^k$,
whose matrix w.r.t.\ the standard bases is real, such that 
$\pi(V_j)=\C^k$ for every irreducible component  $V_j$ of~$V$.
\end{corol}

\begin{proof} 
We construct $\pi$ iteratively by composing standard projections
and appropriate coordinate changes. First we choose
a nonzero polynomial $f$ vanishing on $V$, and we fix
a change of coordinates as in Lemma~\ref{l:coordinate}
so that the corresponding polynomial $f'$ is as in the
projection theorem. Letting $\pi'_d\:\C^d\to\C^{d-1}$
be the composition of the standard
projection $\pi_d$ with
this coordinate change, we get that $\pi'_d(V)$ is a variety
and $\dim\pi'_d(V)=k$. 

Let $V_j$ be an irreducible component of $V$. Then $f$ vanishes
on $V_j$ as well, and applying the projection theorem with $V_j$
instead of $V$, we get that $\pi'_d(V_j)$ is a $k$-dimensional
variety in $\C^{d-1}$ as well. 

We define $\pi'_i\:\C^i\to\C^{i-1}$, $i=d-1,d-2,\ldots,k+1$, 
analogously; to get $\pi'_i$, we use some nonzero
polynomial $f$ that vanishes on the $k$-dimensional
variety $\pi'_{i+1}\circ\cdots\circ\pi'_d(V)$.
The desired projection $\pi$ is the composition
$\pi:=\pi'_{k+1}\circ\cdots\circ\pi'_d$.

We get that $\pi(V)$ is a $k$-dimensional variety in $\C^k$,
and so is $\pi(V_j)$ for every irreducible component $V_j$ of $V$.
But the only $k$-dimensional variety in $\C^k$ is $\C^k$,
and the corollary follows.
\end{proof}

Now we are ready to prove the key lemma.

\begin{proof}[Proof of Lemma \ref{l:key}]
Given the $k$-dimensional complex variety $V$ and the $n$-point set
$Q\subset\R^d$ as in the key lemma, we consider a projection~$\pi\:\C^d\to\C^k$
as in Corollary~\ref{c:thatproj}.


Since the matrix of $\pi$ is real, we can regard $\bar Q:=\pi(Q)$
as a subset of $\R^k$. More precisely, $\bar Q$ is a \emph{multiset}
in general, since $\pi$ may send several points to the same point.
(It would be easy to avoid such coincidences in the choice of
$\pi$, but we do not have to bother with that.)

We apply the original Guth--Katz polynomial partition
theorem to $\bar Q$, which yields a $\frac1r$-partitioning 
polynomial $\bar g\in\R[y_1,\ldots,y_k]$ for $\bar Q$
of degree $D=O(r^{1/k})$. We note that the Guth--Katz method
works for multisets without any change (because the ham-sandwich
theorem used in the proof applies to arbitrary measures and thus,
in particular, to multisets).

We define a polynomial $g\in\R[x_1,\ldots,x_d]$ as the pullback
of $\bar g$, i.e., $g(x):=\bar g(\pi(x))$. We have
$\deg g=\deg\bar g$ since $\pi$ is linear and surjective.

Moreover, $g$ is a $\frac1r$-partitioning polynomial for $Q$,
since if $\pi(q)$ and $\pi(q')$ lie in different components
of $\R^k\setminus Z(\bar g)$, 
then $q$ and $q'$ lie in different
components of $\R^d\setminus Z(g)$ (indeed, if not, a path
$\gamma$ connecting $q$ to $q'$ and avoiding $Z(g)$ would project
to a path $\bar\gamma$ connecting $\pi(q)$ to $\pi(q')$
and avoiding $Z(\bar g)$).

Finally, since $\bar g$ does not vanish identically on $\C^k$
and $\pi(V_j)=\C^k$ for every $j$, the polynomial $g$ does not vanish identically
on any of the irreducible components $V_j$. The key lemma is proved.
\end{proof}

\section{Proof of Theorem~\ref{t:multilev}}\label{s:overview}

Here we use the key lemma to construct the multilevel partition
in Theorem~\ref{t:multilev}. Thus, we are given an $n$-point
set $P\subset\R^d$ and a parameter $r>1$.

We proceed in $d$ steps. The parameters $r_1,r_2,\ldots,r_d$
are set as follows:
\[
r_1:=r,\ \ r_{i+1}:=r_i^c,\ \ i=1,2,\ldots,d-1,
\]
where $c$ is a sufficiently large constant (depending on $d$).
This  will allow us to consider quantities depending polynomially on $r_i$
as very small compared to~$r_{i+1}$.
We will also have auxiliary degree parameters $D_1,D_2,\ldots,D_d$,
where
\[
D_i=O\left(r_i^{1/(d-i+1)}\right).
\]

At the beginning of the $i$th step, $i=1,2,\ldots,d$,
we will have the following objects:
\begin{itemize} 
\item A complex variety $V_{i-1}$, which may be reducible,
but such that all irreducible components have
 dimension $d-i+1$. Initially, for $i=1$,
$V_0=\C^d$.
\item A set $Q_{i-1}\subseteq P\cap V_{i-1}$, the current
``exceptional set'' that still needs to be partitioned. For $i=1$,
$Q_{0}=P$.
\end{itemize}

We also have
\[
\deg V_{i-1}\le \Delta_{i-1}:= D_1D_2\cdots D_{i-1}.
\]

In the $i$th step, we apply the key lemma to $V_{i-1}$ and $Q_{i-1}$
with $r=r_i$ (and $k=d-i+1$). This yields a real $(1/r_i)$-partitioning
polynomial $g_i$ for $Q_{i-1}$ of degree at most $D_i=O\left(r_i^{1/(d-i+1)}\right)$
that does not vanish identically on any of the irreducible components
of $V_{i-1}$.
(For $i=1$, this is just an application of the original Guth--Katz
polynomial partition theorem.)

Let $S_{i1},\ldots,S_{it_i}$ be the connected components of 
$(V_{i-1}\cap\R^d)\setminus Z(g_i)$, and let $P_{ij}:=S_{ij}\cap Q_{i-1}$
(these are the sets as in Theorem~\ref{t:multilev}).
For every $j$ we  have $|P_{ij}|\le |Q_{i-1}|/r_i\le n/r_i$ since
$g_i$ is a $(1/r_i)$-partitioning polynomial. 
We also have the new exceptional set $Q_{i}:= Q_{i-1}\cap Z(g_i)$.

Finally, we set $V_i:= V_{i-1}\cap Z_\C(g_i)$.
Since $g_i$ does not vanish identically on any of the irreducible
components of $V_{i-1}$, all irreducible components
of $V_{i}$ are $(d-i)$-dimensional by Lemma~\ref{l:gbezout},
and the sum of their degrees, which equals
$\deg V_i$,  is at most 
\[\deg(V_{i-1})\deg(g_i)\le
\Delta_{i-1}D_i
\le D_1D_2\cdots D_{i}=\Delta_i
\] as needed for the next inductive step.
This finishes the $i$th partitioning step.

After the $d$th step, we end up with a $0$-dimensional variety $V_d$,
whose irreducible components are points, 
and their number is $\deg V_d\le\Delta_d$, a quantity polynomially
bounded in~$r$. The set $Q_d$ is
 the exceptional set $P^*$ in Theorem~\ref{t:multilev},
and $|Q_d|\le|V_d|=\deg V_d\le\Delta_d$.

\heading{The crossing number. } It remains to prove the bounds
on the number of the sets $S_{ij}$ crossed by $X$ as in parts (i)
and (ii) of the theorem.

First let $X=Z(h)$ be a hypersurface of degree $D_0=O(1)$ as in (i).
For $i=1$, we actually get that $X$ \emph{intersects}
at most $O\left(r_1^{1-1/d}\right)$ of the $S_{1j}$, because the number of
the $S_{1j}$ intersected by $X$ is no larger than the number
of connected components of $X\setminus Z(g_1)$.
By the Barone--Basu theorem (Theorem~\ref{t:basu}),
the number of these components is bounded by $O((\deg h)(\deg g_1)^{d-1})=
O\left(D_0 D_1^{d-1}\right)=O\left(r_1^{1-1/d}\right)$ as claimed.

Now let $i\ge 2$. We want to bound
the number of the sets $S_{ij}$ crossed by $X$. 
Let $U_1,\ldots,U_b$ be the irreducible components of $V_{i-1}$
whose real points are not completely contained in $X$;
that is, satisfying $U_\ell\cap\R^d\not\subseteq X$.
We have $b\le\deg V_{i-1}\le\Delta_{i-1}$.

For every $j$ such that $X$ crosses $S_{ij}$, let us fix a point
$y_j\in S_{ij}\setminus X$ and another point $z_j\in S_{ij}\cap X$ 
(they exist by the definition of crossing). Since $S_{ij}$
is path-connected, there is also a path $\gamma_j\subseteq S_{ij}$ 
connecting $y_j$ to $z_j$.

Let $z^*_j$ be the first point of $X$ on $\gamma_j$ when we go from
$y_j$ towards $z_j$. We observe that $z^*_j$ lies in some $U_\ell$.
Indeed, points on $\gamma_j$ just before $z^*_j$ lie in $V_{i-1}$
(since $S_{ij}\subseteq V_{i-1}$) but not in $X$, hence they lie
in some $U_\ell$, and $U_\ell$, being an algebraic variety, is closed
in the Euclidean topology.

For any given $U_\ell$, a connected component of 
$(U_\ell\cap\R^d\cap X)\setminus Z(g_i)$ may contain at most one of the
 $z^*_j$ (since the $S_{ij}$ are separated by $Z(g_i)$). Therefore,
the number of the $S_{ij}$ crossed by $X$ is no more than
\[
\sum_{\ell=1}^b \#(W_\ell\setminus Z(g_i)),
\]
 where  $W_\ell:=U_\ell\cap\R^d\cap X$, 
and $\#$ denotes the number of connected components.

Since $U_\ell$ is irreducible and $X$ does not contain all of its
real points, the polynomial $h$ defining $X$ does not vanish on $U_\ell$,
and thus $U_\ell\cap Z_\C(h)$
is a proper subvariety of $U_\ell$ of (complex) dimension $\dim U_\ell-1=d-i$.
Hence, by Lemma \ref{l:dims}, the real variety $W_\ell=(U_\ell\cap Z_\C(h))
\cap\R^d$ also has (real) dimension at most $d-i$.

By Theorem~\ref{t:heintz-def}, we have $U_\ell=Z_\C(f_1,\ldots,f_m)$
for some, generally complex, polynomials of degree at most $\deg U_\ell\le
\Delta_{i-1}$.
Thus $W_\ell$ is the  real zero set of
the real polynomials $h$, $f_1\overline{f}_1,\ldots,f_m\overline{f}_m$.
These polynomials
have degrees bounded by $\max(D_0,2\Delta_{i-1})=O(\Delta_{i-1})$.

By the Barone--Basu theorem again, the number of components
of $W_\ell\setminus Z(g_i)$ is at most
\[
O(\Delta_{i-1}^{d-\dim W_\ell} D_i^{\dim W_\ell})=
O(\Delta_{i-1}^d D_i^{d-i})=O\left(\Delta_{i-1}^d r_i^{1-1/(d-i+1)}\right).
\]
The total number of the $S_{ij}$ crossed by $X$ is then bounded by
$\Delta_{i-1}$ times the last quantity, i.e., by
$O\left(\Delta_{i-1}^{d+1}r_i^{1-1/(d-i+1)}\right)=
O\left((D_1D_2\cdots D_{i-1})^{d+1}r_i^{1-1/(d-i+1)}\right)$.
Since $r_i=r_{i-1}^c$, we can make $(D_1D_2\cdots D_{i-1})^{d+1}$
smaller than any fixed power of $r_i$, 
and hence we can bound the last estimate by $O\left(r_i^{1-1/d}\right)$
(recall that $i\ge 2$),
which finishes the proof of part (i) of the theorem.

For part (ii), the argument requires only minor modifications.
Now $X$ is a variety of dimension $k\le d-2$ defined by real
polynomials of degree at most $D_0=O(1)$. 

We have $\dim V_{i-1}=d-i+1$, and for $\dim X=k\le d-i$ we simply
count the components of $X\setminus Z(g_i)$, as we did for part (i)
in the case $i=1$. This time we obtain the bound
$O\left(D_0^{d-\dim X}D_i^{\dim X}\right)=O\left(D_i^k\right)=O\left(r_i^{k/(d-i+1)}\right)$.

The exponent $\frac k{d-i+1}$ increases with $i$, and
thus it is the largest for $d-i=k$, in which case
the bound is $O\left(r_i^{1-1/(k+1)}\right)$. (This is the critical case;
for all of the other $i$ we get a better bound.)

For $k\ge d-i+1$, we argue as in part (i): letting
$U_1,\ldots,U_b$ be the irreducible components of $V_{i-1}$
with $U_\ell\cap\R^d\not\subseteq X$ and $W_\ell:=
U_\ell\cap\R^d\cap X$, the number of the $S_{ij}$ crossed
by $X$ is bounded by $\sum_{\ell=1}^b \#(W_\ell\setminus Z(g_i))$,
and each $W_\ell$ has (real) dimension at most $\dim V_{i-1}-1=d-i$.
The number of components of $W_\ell\setminus Z(g_i)$ is again
bounded, by the Barone--Basu theorem, by
$O\left(\Delta_{i-1}^d r_i^{1-1/(d-i+1)}\right)$, and the sum over
all $W_\ell $ is $O\left(\Delta_{i-1}^{d+1} r_i^{1-1/(d-i+1)}\right)$.
For every fixed $\delta>0$, we can choose 
the constant $c$ in the inductive definition
of the $r_i$ so large that
$\Delta_{i-1}^{d+1}\le r_i^\delta$,
and so the previous bound is
no more than $O\left(r_i^{1-1/(d-i+1)+\delta}\right)$.

The exponent $1-\frac{1}{d-i+1}$ 
is maximum for $d-i+1=k$, in which case our bound
is $O\left(r_i^{1-1/k+\delta}\right)$. By letting $\delta:=
\frac 1k-\frac1{k+1}$, we bound this by
$O\left(r_i^{1-1/(k+1)}\right)$.
This concludes the proof of
Theorem~\ref{t:multilev}.

\section{Algorithmic aspects of Theorem~\ref{t:multilev}}\label{s:algo}

 The goal of this section is to prove Theorem \ref{t:algo-multilev}. In order to make the proof of Theorem~\ref{t:multilev} algorithmic,
we need to compute both with real and complex varieties. 
A variety $V$, both in the real and complex cases,
is represented by a finite list $f_1,\ldots,f_m$
of polynomials such that $V=Z(f_1,\ldots,f_m)$.

The size of such a representation is measured as $m+\sum_{i=1}^m\deg f_i$.
It would perhaps be more adequate to use ${\deg f_i+d\choose d}$, the number
of monomials in a general $d$-variate polynomial of degree $\deg f_i$,
instead of just $\deg f_i$, but since we consider $d$ constant,
both quantities are polynomially equivalent.

If we want to pass from a complex $V$ defined by generally complex
polynomials $f_1,\ldots,f_m$
to the real variety $V\cap\R^d$, we use the trick already 
mentioned: $V\cap\R^d$ is defined by the real polynomials $
f_1\overline{f}_1,\ldots, f_m\overline{f}_m$.

To make the construction in Theorem~\ref{t:multilev} algorithmic,
besides some obvious steps (such as testing the membership of a point
in a variety, which is done by substituting the point coordinates
into the defining polynomials), we need to implement the following operations:
\begin{enumerate}
\item[(A)] Given a variety $V$ in $\C^d$ of dimension $k$, $1\leq k\le d-1$,
such that all irreducible components of $V$ have dimension $k$,
compute a real projection $\pi\:\C^d\to\C^k$ as in Corollary~\ref{c:thatproj},
i.e., such that $\pi(V_j)=\C^k$ for all irreducible components
$V_j$ of~$V$ .
\item[(B)]  Given a point (multi)set $Q\subset\R^k$, $k\le d$,
 construct  a $\frac1r$-partitioning polynomial
of degree $O\left(r^{1/k}\right)$ (as in the proof of the key lemma).
\item[(C)] Given a complex variety $V$ and a polynomial
$g$, compute $V\cap Z_\C(g)$.
\end{enumerate}

For (A),  we follow the proof of Corollary~\ref{c:thatproj},
i.e., we compute $\pi$ as the composition 
$\pi'_{k+1}\circ\cdots\circ\pi'_{d}$, where $\pi'_i\:\C^i\to\C^{i-1}$
sends $(x_1,\ldots,x_i)$ to $(x_1+\lambda_{i,1} x_i,\ldots,x_{i-1}+
\lambda_{i,i-1} x_i)$, with the $\lambda_{ij}$ chosen independently
at random from the uniform distribution on $[0,1]$, say
(or, if we do not want to assume the capability
of generating such random reals, 
 we can still choose them as random integers in a sufficiently
large range).
The composed $\pi$ will work almost surely
(or, if we use large random integers, with high probability---this
can be checked using the Schwartz--Zippel lemma).

In order to verify that a particular $\pi$ works, we verify the condition
in the projection theorem (Theorem~\ref{t:proj}) 
for each $\pi'_i$ separately. To this end, we 
compute the projected varieties $V_i:=\pi'_{i+1}\circ\cdots
\circ\pi'_d(V)$ in $\C^i$; initially $V_d=V$. 

The projections can be computed in a standard way using Gr\"obner
bases w.r.t.\ the lexicographic ordering; see \cite{CLO}. Namely,
we suppose that $V_i$ has already been computed.
We make the substitution
$x'_j:= x_j+\lambda_{ij}x_i$, where the $\lambda_{ij}$ are 
those used in $\pi^{(i)}$ and $\lambda_{ii}=0$; this transforms the
list of polynomials defining $V_i$ into another list of polynomials
in the new variables $x'_1,\ldots,x'_i$.  Since $1\leq\dim V_i\leq d-1$, it follows that all the polynomials in the list have
degree at least one. Thus, by Theorem \ref{t:groebner}, we compute a 
Gr\"obner basis $G_i$ of the ideal generated by these new polynomials,
with respect to the lexicographic ordering, where the ordering puts
the variable $x_i$ first.

If $G_i$ contains no polynomial whose leading term is a power of $x_i$
(as in the projection theorem),
then we discard $\pi_i$, generate a new one, and repeat the test.
If $G_i$ does contain such a polynomial, then we take all polynomials
in $G_i$ that do not contain $x_i$, and these define the variety
$V_{i-1}=\pi'_i(V_i)$ in $\C^{i-1}$.
Indeed, recall that by \cite[Thm. 3.1.2]{CLO}, if $G$ is a Gr\"obner basis of $I \subseteq \C[x_1,\ldots,x_d]$ then $G\cap \C[x_1,\ldots,x_{d-1}]$
is a Gr\"obner basis of $I\cap \C[x_1,\ldots,x_{d-1}].$ The claim now follows from the projection theorem.

Thus, the computation of $\pi$ takes a constant number of Gr\"obner
basis computations and the expected number of repetitions is a constant.
(In practice, the coordinate projection forgetting the last $d-k$
coordinates will probably work most of the time; then only one
Gr\"obner basis computation is needed to verify that it works.)

For operation (B), constructing a partitioning polynomial 
for points in $\R^k$,
we use a (randomized) algorithm from \cite[Thm.~1.1]{AMSII},
which runs in expected time $O\left(|Q|r+r^3\right)$ for fixed $k$.
It also works for multisets, as can easily be checked.
Since each point of the original input set $P$ participates in no more
than $d$ of these operations, and the value of $r$ in each of these
cases is bounded by a polynomial function of the original
parameter $r$ in the theorem,
the total time spent in all of the operations (B) in the construction
is bounded by $O(nr^C)$ for a constant~$C$.

Operation (C), intersecting a complex variety
with $Z(g)$, is trivial in our representation, since we just add $g$
to the list of the defining polynomials of~$V$.

This finishes the implementation of the operations, and now
we need to substantiate the claims about the number and form
of the sets $S_{ij}$. We recall that each $S_{ij}$ is obtained
as a cell in the arrangement of
$Z(g_i)$ within~$V_{i-1}$. The degrees of
$g_i$ and of the polynomials defining $V_{i-1}$ 
are bounded by a polynomial in $r$.
Then by Theorem~\ref{t:make-arrg},
we get that each $S_{ij}$ is defined by at most $r^C$
polynomials of degree at most $r^C$, and is computed in $r^C$ 
time. The number of the $S_{ij}$ is polynomially bounded 
in $r$ as well.

Finally, we need to consider a range $\gamma\in\Gamma_{d,D_0,s}$.
By definition, $\gamma$ is
a Boolean combination of $\gamma_1,\ldots,\gamma_s$,
where $\gamma_\ell=\{x\in\R^d: h_\ell(x)\ge 0\}$, with a polynomial
$h_\ell$ of degree at most $D_0$, and moreover, if $\gamma$ crosses
a path-connected set $A$, then at least one of the varieties $X_\ell=Z(h_\ell)$
crosses~$A$. It follows that the crossing number for $\gamma$
is no more than $s$-times the bound in Theorem~\ref{t:multilev}(i).
This concludes the proof
of Theorem~\ref{t:algo-multilev}.

\section{The range searching result}\label{s:rgs}

The derivation of the range searching result, Theorem~\ref{t:large-r},
from Theorem~\ref{t:algo-multilev}, is by a standard construction
of a partition tree as in \cite{m-ept-92,AMSII},
 and here we give it for completeness (and also to illustrate
its simplicity).

\begin{proof}[Proof of Theorem~\ref{t:large-r}.]
Given $d,D_0,s$, $\eps>0$ and a set $P\subset\R^d$, we choose
a sufficiently large $n_0=n_0(d,D_0,s,\eps)$ and a sufficiently small parameter
$\eta=\eta(d,D_0,s,\eps)>0$, and we construct a partition
tree $\TT$ for $P$ recursively as follows:

If $|P|\le n_0$, $\TT$ consists of a single node storing
a list of the points of $P$ and their weights.

For $|P|>n_0$, we choose $r:=n^\eta$ and we construct
$P^*$, the $P_{ij}$, and the $S_{ij}$ as in Theorem~\ref{t:multilev}.
The root of $\TT$ stores (the formulas defining) the $S_{ij}$,
the total weight of each $P_{ij}$, and the points of $P^*$ 
together with their weight. For each $i$ and $j$, we 
make a subtree of the root node, which is a partition tree
for $P_{ij}$ constructed recursively by the same method.

By Theorem~\ref{t:algo-multilev}, the construction of the root node
of $\TT$ takes expected time $O\left(nr^C\right)=O\left(n^{1+C\eta}\right)$.
The total preprocessing time $T(n)$ for an $n$-point $P$
obeys the recursion, for $n>n_0$,
 $T(n)\le O\left(n^{1+C\eta}\right)+\sum_{i,j}T(n_{ij})$, with $\sum_{i,j}n_{ij}\le n$
and $n_{ij}\le n/r=n^{1-\eta}$, whose solution is $T(n)\le O\left(n^{1+C\eta}\right)$.
A similar simple analysis shows that the total storage requirement
is $O(n)$.

Let us consider answering a query with  a query range 
$\gamma\in\Gamma_{d,D_0,s}$. 
We start at the root of $\TT$ and maintain a global counter
which is initially set to $0$.
We test the points of the exceptional
set $P^*$ for membership in $\gamma$ one by one and increment
the counter accordingly in $r^{O(1)}$ time. 
Then, for each $i,j$, we distinguish three possibilities:
\begin{enumerate}
\item[(i)] If $S_{ij}\cap\gamma=\emptyset$, we do nothing.
\item[(ii)] If $S_{ij}\subseteq\gamma$, we add the total weight of the
points of $P_{ij}$ to the global counter.
\item[(iii)] Otherwise, we recurse in the subtree corresponding to $P_{ij}$,
which increments the counter by the total weight of the points
of $P_{ij}\cap\gamma$.
\end{enumerate}

The three possibilities above can be distinguished, for given $S_{ij}$,
by constructing the arrangement of the zero sets of the 
polynomials defining $S_{ij}$ plus  the polynomials
defining $\gamma$, according
to Theorem~\ref{t:make-arrg}. The total time, for all $i,j$ together,
is $r^{O(1)}$. 

Since, by Theorem~\ref{t:algo-multilev}, $\gamma$ together crosses
at most $O\left(r_i^{1-1/d}\right)$ of the $S_{ij}$, possibility (iii)
occurs, for given $i$,
 for at most $O\left(r_i^{1-1/d}\right)$ values of $j$. We thus obtain the following
recursion for the query time $Q(n)$, with the initial condition
$Q(n)=O(1)$ for $n\le n_0$:
\[
Q(n)\le n^{C'\eta} + \sum_{i=1}^d O\left(r_i^{1-1/d}\right)Q\left(n/r_i\right),
\ \ \ \ n^\eta\le r_i\le n^{K\eta},
\]
where $C'$ and $K$ are constants independent of $\eta$.
A simple induction on $n$ verifies that this implies,
for $\eta\le (1-1/d)/C'$,
$Q(n)=O\left(n^{1-1/d}\log^B n\right)$ as claimed.
\end{proof}

\section{Remark: On (not) computing irreducible components}\label{s:irred}

For the algorithmic part, it is important that we do not need
to compute the irreducible components of the varieties $V_i$
(although we use the irreducible components in the proof
of our multilevel partition theorem).

There are several algorithms
in the literature for computing irreducible components
of a given complex variety (e.g.,  \cite{ElkadiMourrain}).
However, these algorithms need factorization of multivariate
polynomials over $\C$ as a subroutine (after all, factoring
a polynomial corresponds to computing irreducible components
of a hypersurface).

Polynomial factorization is a well-studied
topic, with many impressive results; see, e.g.,
\cite{kaltofen-surv91} for a survey. In particular, there
are algorithms that work in polynomial time, assuming the
dimension fixed, but only in the Turing machine model. Adapting
these algorithms to the Real RAM model, which is common in
computational geometry and which we use, encounters some
nontrivial obstacles---we are grateful to Erich Kaltofen
for explaining this issue to us.

It may perhaps be possible
to overcome these obstacles by techniques used in real algebraic
geometry for computing in abstract real-closed fields
(see \cite{BasuPollackRoy-book}), but this would need to be worked
out carefully. Then one could probably obtain rigorous complexity
bounds on computing irreducible components of a complex variety,
hopefully polynomial in fixed dimension; we find this question
of independent interest. 

%

\subsection*{Acknowledgment}

We would like to thank Josh Zahl for pointing out mistakes
in an earlier version of this paper, Saugata Basu for
providing a draft of his recent work with Sombra and
useful advice, Erich Kaltofen for kindly answering our questions
concerning polynomial factorization,
and Pavel Pat\'ak, Edgardo Rold\'an Pensado, Mart\'in Sombra, and Martin Tancer for enlightening discussions.

\bibliographystyle{alpha}
\bibliography{pp,cg,../geom}
\end{document}